\documentclass[]{interact}

\usepackage{epstopdf}
\usepackage[caption=false]{subfig}

\usepackage[numbers,sort&compress]{natbib}

\usepackage{url}
\bibpunct[, ]{[}{]}{,}{n}{,}{,}
\renewcommand\bibfont{\fontsize{10}{12}\selectfont}
\makeatletter
\def\NAT@def@citea{\def\@citea{\NAT@separator}}
\makeatother

\theoremstyle{plain}
\newtheorem{theorem}{Theorem}[section]

\newtheorem{corollary}[theorem]{Corollary}
\newtheorem{proposition}[theorem]{Proposition}

\theoremstyle{definition}
\newtheorem{definition}[theorem]{Definition}
\newtheorem{example}[theorem]{Example}

\theoremstyle{remark}
\newtheorem{remark}{Remark}

\usepackage{amsmath}
\usepackage{amssymb}
\usepackage{mathrsfs}
\usepackage{mathtools}
\usepackage{bm}
\usepackage{graphicx}

\newcommand{\defref}[1]{Definition~\ref{#1}}

\newcommand{\exref}[1]{Example~\ref{#1}}
\newcommand{\secref}[1]{Section~\ref{#1}}
\newcommand{\figref}[1]{Fig.~\ref{#1}}

\newcommand{\R}{\mathbb{R}}
\newcommand{\N}{\mathbb{N}}
\newcommand{\Z}{\mathbb{Z}}
\newcommand{\G}{\mathbb{G}}
\newcommand{\X}{\mathcal{X}}
\newcommand{\K}{\mathcal{K}}
\newcommand{\M}{\mathcal{M}}
\newcommand{\C}{\mathcal{C}}
\newcommand{\KeyGen}{\mathsf{KeyGen}}
\newcommand{\Enc}{\mathsf{Enc}}
\newcommand{\Dec}{\mathsf{Dec}}
\newcommand{\Eval}{\mathsf{Eval}}
\newcommand{\KeyUpd}{\mathsf{KeyUpd}}
\newcommand{\CtUpd}{\mathsf{CtUpd}}
\newcommand{\EC}{\mathsf{EC}}
\newcommand{\Ecd}{\mathsf{Ecd}}
\newcommand{\Dcd}{\mathsf{Dcd}}
\newcommand{\Sum}{\mathsf{Sum}}
\newcommand{\pk}{\mathsf{pk}}
\newcommand{\sk}{\mathsf{sk}}
\newcommand{\ct}{\mathsf{ct}}
\newcommand{\tr}{\mathop{\mathrm{tr}}\limits}

\renewcommand{\vec}{\mathop{\mathrm{vec}}\limits}
\newcommand{\argmin}{\mathop{\mathrm{arg~min}}\limits}

\newcommand{\EV}{\mathop{\mathbb{E}}\limits}

\allowdisplaybreaks[4]

\begin{document}

\articletype{ARTICLE TEMPLATE}

\title{Optimal Security Parameter for Encrypted Control Systems Against Eavesdropper and Malicious Server}

\author{
\name{Kaoru Teranishi\textsuperscript{a,b}\thanks{CONTACT Kaoru Teranishi. Email: teranishi@uec.ac.jp} and Kiminao Kogiso\textsuperscript{a}}
\affil{\textsuperscript{a}Department of Mechanical and Intelligent Systems Engineering, The University of Electro-Communications, 1-5-1 Chofugaoka, Chofu, Tokyo 1828585, Japan; \textsuperscript{b}Research Fellow of Japan Society for the Promotion of Science, Kojimachi Business Center Building, 5-3-1 Kojimachi, Chiyoda-ku, Tokyo 1020083, Japan}
}

\maketitle

\begin{abstract}
A sample identifying complexity and a sample deciphering time have been introduced in a previous study to capture an estimation error and a computation time of system identification by adversaries.
The quantities play a crucial role in defining the security of encrypted control systems and designing a security parameter.
This study proposes an optimal security parameter for an encrypted control system under a network eavesdropper and a malicious controller server who attempt to identify system parameters using a least squares method.
The security parameter design is achieved based on a modification of conventional homomorphic encryption for improving a sample deciphering time and a novel sample identifying complexity, characterized by controllability Gramians and the variance ratio of identification input to system noise.
The effectiveness of the proposed design method for a security parameter is demonstrated through numerical simulations.
\end{abstract}


\section{Introduction}
\label{sec:introduction}

Outsourcing computation of controllers to a cloud server, such as control as a service (CaaS), is one form of realization of cyber-physical systems that improve the efficiency and flexibility of traditional control systems.
However, such computing services often face threats that adversaries eavesdrop and learn about private information of control systems.
Homomorphic encryption is the major countermeasure against such threats because it provides direct computation on encrypted data without accessing the original messages~\cite{acar2019}.
The encryption was applied to realize an encrypted control that is a framework for secure outsourcing computation of control algorithms~\cite{kogiso2015,farokhi2017,kim2016,kim2022,darup2021}.
Owning to the benefits of encrypted control, various controls, such as model predictive control~\cite{alexandru2018,darup2018a}, motion control~\cite{qiu2019,shono2022}, and reinforcement learning~\cite{suh2021a}, were implemented in encrypted forms.

Some recent studies have defined and analyzed the security of encrypted control systems through two approaches to clarify how secure an encrypted control system is against what type of adversary.
One of them is a cryptographic approach that defines the provable security of encrypted controls and reveals a relation between the security and existing security notions in cryptography~\cite{teranishi2022}.
In this security definition, an adversary and information used for attacks are formulated as a probabilistic polynomial-time algorithm and its inputs, respectively, instead of assuming specific attacks.
Using the security notion, we can analyze qualitative security for a broad class of encrypted control systems.
In contrast, other studies employed a control theoretic approach that considers the security of encrypted control systems under an adversary who wants to learn the system parameters by system identification~\cite{teranishi2022a,teranishi2023}.
The security in this approach is defined by the system identification error and computation time for the process.
Unlike the cryptographic approach, the security notion in this approach enables quantifying a security level of encrypted control systems.
The studies also solved an optimization problem for designing a security parameter to minimize the computation costs of encryption algorithms while satisfying the desired security level.

This study focuses on designing an optimal security parameter for encrypted control systems under an adversary who attempts to identify the system and input matrices of a system controlled by an encrypted controller, although the conventional works~\cite{teranishi2022a,teranishi2023} dealt with an adversary identifying a system matrix of a closed-loop system.
Such an adversary represents a network eavesdropper executing man-in-the-middle attacks and a malicious controller server infected by malware or spoofing an authorized server computing encrypted control algorithms.
Furthermore, the adversary employs a basic least squares identification method, which is more prevalent in practical use than the Bayesian estimation method discussed in~\cite{teranishi2022a}.

Unfortunately, the existing design methods for an optimal security parameter are effective only against a network eavesdropper.
That is, they cannot work for a malicious controller server appropriately.
The existing methods must share a token in updatable homomorphic encryption, of which key pairs are updated every sampling period, with a controller server to update controller ciphertexts.
Furthermore, the update token needs to be kept secret against adversaries because it can be exploited to estimate past and future key pairs from the current key pair.
Indeed, the previous study~\cite{teranishi2022a} assumed that an update token is transmitted by a secure communication channel using traditional symmetric-key encryption, such as AES.
However, such an assumption is not valid for a malicious controller server because the ciphertext of an update token must be decrypted on the server.
Hence, the design of an optimal security parameter for encrypted control systems is still a challenging problem when an adversary is a malicious controller server rather than a network eavesdropper.

To solve the problem, this study modifies the updatable homomorphic encryption in~\cite{teranishi2022a}.
The modified encryption enables the computation of encrypted data and correct decryption without sharing an update token while updating key pairs.
Furthermore, we propose a novel sample identifying complexity, which is characterized by controllability Gramians and variance ratio of adversarial input for the system identification and system noise, for defining the security of encrypted control systems under the eavesdropper and malicious server.
Using the proposed complexity, we can estimate how precisely the adversaries are expected to identify the system and input matrices of a given system for a certain number of data.
We design an optimal security parameter for an encrypted control system under the adversaries using the proposed updatable homomorphic encryption and sample identifying complexity.

The rest of this paper is organized as follows.
\secref{sec:preliminaries} defines the syntax and security of homomorphic encryption and encrypted control.
\secref{sec:threat} formulates a threat model considered in this study.
\secref{sec:updatable_homomorphic_encryption} presents a modified homomorphic encryption.
\secref{sec:security_parameter_design} proposes a novel sample identifying complexity and an optimal security parameter for the modified encryption.
\secref{sec:simulation} shows the results of numerical simulations.
\secref{sec:conclusion} describes the conclusions and future work.

\section{Preliminaries}
\label{sec:preliminaries}

\subsection{Notation}
\label{sec:notation}

The sets of natural numbers, integers, and real numbers are denoted by $\N$, $\Z$, and $\R$, respectively.
Key, plaintext, and ciphertext spaces are denoted by $\K$, $\M$, and $\C$, respectively.
Define the set $\Z^+ \coloneqq \{z \in \Z \mid 0 \le z\}$ and a bounded set $\X \subset \R$.
The sets of $n$-dimensional vectors and $m$-by-$n$ matrices of which elements and entries belong to a set $\mathcal{A}$ are denoted by $\mathcal{A}^n$ and $\mathcal{A}^{m \times n}$, respectively.
The $i$th element of a vector $v \in \mathcal{A}^n$ and the $(i,j)$ entry of a matrix $M \in \mathcal{A}$ are denoted by $v_i$ and $M_{ij}$, respectively.
The Euclidean norm and the Frobenius norm of $v \in \mathcal{A}^n$ and $M \in \mathcal{A}^{m \times n}$ are denoted by $\|v\|_2$ and $\|M\|_F$, respectively.
The column stack vector of $M$ is defined as $\vec(M) \coloneqq [M_1^\top \, \cdots \, M_n^\top]^\top$, where $M_i$ is the $i$th column vector of $M$.

\subsection{Homomorphic encryption}
\label{sec:he}

This section introduces the syntax and security level of homomorphic encryption.
First, the syntax of homomorphic encryption~\cite{acar2019} is defined as follows.

\begin{definition}
\label{def:he}
    Homomorphic encryption is $(\KeyGen, \Enc, \Dec, \Eval)$ such that:
    \begin{itemize}
        \item $(\pk, \sk) \gets \KeyGen(1^\lambda)$: A key generation algorithm takes $1^\lambda$ as input and outputs a key pair $(\pk, \sk) \in \K$, where $1^\lambda$ is the unary representation of a security parameter $\lambda \in \N$, $\pk$ is a public key, and $\sk$ is a secret key.
        \item $\ct \gets \Enc(\pk, m)$: An encryption algorithm takes a public key $\pk$ and a plaintext $m \in \M$ as input and outputs a ciphertext $\ct \in \C$.
        \item $m \gets \Dec(\sk, \ct)$: A decryption algorithm takes a secret key $\sk$ and a ciphertext $\ct \in \C$ as input and outputs a plaintext $m \in \M$.
        \item $\ct \gets \Eval(\pk, \ct_1, \ct_2)$: A homomorphic evaluation algorithm takes a public key $\pk$ and ciphertexts $\ct_1, \ct_2 \in \C$ as input and outputs a ciphertext $\ct \in \C$.
        \item Correctness: $\Dec(\sk, \Enc(\pk, m)) = m$ holds for any $(\pk, \sk) \gets \KeyGen(1^\lambda)$ and for any $m \in \M$.
        \item Homomorphism: $\Dec(\sk, \Eval(\pk, \ct_1, \ct_2)) = m_1 \bullet m_2$ holds for any $(\pk, \sk) \gets \KeyGen(1^\lambda)$ and for any $m_1, m_2 \in \M$, where $\ct_1 \gets \Enc(\pk, m_1)$, $\ct_2 \gets \Enc(\pk, m_2)$, and $\bullet$ is a binary operation on $\M$.
    \end{itemize}
\end{definition}

\begin{example}
\label{ex:elgamal}
    The algorithms of ElGamal encryption~\cite{ElGamal85} are as follows.
    \begin{itemize}
        \item $(\pk, \sk) \gets \KeyGen(1^\lambda)$:
        Randomly generate prime numbers $q = q(\lambda)$ and $p = p(\lambda)$ such that $p = nq + 1$ and $n \in \N$.
        Randomly choose $s \in \Z_q$.
        Output $(\pk, \sk) = ((p, q, g, g^s \bmod p), s)$.
        Plaintext and ciphertext spaces are $\M = \G = \{g^i \bmod p \mid i \in \Z_q\}$ and $\C = \G^2$, respectively, where $g^q \mod p = 1$.
        \item $\ct \gets \Enc(\pk, m)$:
        Parse $\pk = (p, q, g, h)$.
        Randomly choose $r \in \Z_q$.
        Output $\ct = (g^r \bmod p, mh^r \bmod p)$.
        \item $m \gets \Dec(\sk, \ct)$:
        Parse $\ct = (c_1, c_2)$.
        Set $s = \sk$.
        Output $m = c_1^{-s}c_2 \bmod p$.
        \item $\ct \gets \Eval(\pk, \ct_1, \ct_2)$:
        Parse $\pk = (p, q, g, h)$, $\ct_1 = (c_{11}, c_{12})$, and $\ct_2 = (c_{21}, c_{22})$.
        Output $\ct = (c_{11}c_{21} \bmod p, c_{12}c_{22} \bmod p)$.
    \end{itemize}
    The ElGamal encryption is multiplicative homomorphic encryption, i.e., $\Dec(\sk, \Eval(\pk, \ct_1, \ct_2)) = m_1 m_2 \bmod p$.
\end{example}

Next, updatable homomorphic encryption~\cite{teranishi2023} is defined as follows.

\begin{definition}
\label{def:upd_he}
    Let $\Pi=(\KeyGen, \Enc, \Dec, \Eval)$ be homomorphic encryption.
    Updatable homomorphic encryption is $(\Pi, \KeyUpd, \CtUpd)$ such that:
    \begin{itemize}
        \item $(\pk_{t+1}, \sk_{t+1}, \sigma_t) \gets \KeyUpd(\pk_t, \sk_t)$: A key update algorithm takes a key pair $(\pk_t, \sk_t) \in \K$ at time $t \in \Z^+$ as input and outputs an updated key pair $(\pk_{t+1}, \sk_{t+1}) \in \K$ and an update token $\sigma_t$.
        \item $\ct_{t+1} \gets \CtUpd(\ct_t, \sigma_t)$: A ciphertext update algorithm takes a ciphertext $\ct_t \in \C$ and an update token $\sigma_t$ at time $t \in \Z^+$ as input and outputs an updated ciphertext $\ct_{t+1} \in \C$.
        \item Correctness: $\Dec(\sk_t, \ct_t) = \Dec(\sk_t, \Enc(\pk_t, m)) = m$ holds for any $(\pk_0, \sk_0) \gets \KeyGen(1^\lambda)$, for any $m \in \M$, and for all $t \in \Z^+$, where $\ct_0 \gets \Enc(\pk_0, m)$, $(\pk_{t+1}, \sk_{t+1}, \sigma_t) \gets \KeyUpd(\pk_t, \sk_t)$, and $\ct_{t+1} \gets \CtUpd(\ct_t, \sigma_t)$.
        \item Homomorphism: $\Dec(\sk_t, \Eval(\pk_t, \ct_{1,t}, \ct_{2,t})) = \Dec(\sk_t, \Eval(\pk_t, \Enc(\pk_t, m_1), \Enc(\pk_t, m_2))) = m_1 \bullet m_2$ holds for any $(\pk_0, \sk_0) \gets \KeyGen(1^\lambda)$, for any $m_i \in \M$, and for all $t \in \Z^+$, where $\ct_{i,0} \gets \Enc(\pk_0, m_i)$, $(\pk_{t+1}, \sk_{t+1}, \sigma_t) \gets \KeyUpd(\pk_t, \sk_t)$, $\ct_{i,t+1} \gets \CtUpd(\ct_{i,t}, \sigma_t)$, and $i = 1, 2$.
    \end{itemize}
\end{definition}

\begin{example}
\label{ex:dyn_elgamal}
    The algorithms of dynamic-key ElGamal encryption~\cite{teranishi2022a} are as follows.
    \begin{itemize}
        \item The key generation, encryption, decryption, and homomorphic evaluation algorithms are identical to the ElGamal encryption in \exref{ex:elgamal}.
        \item $(\pk_{t+1}, \sk_{t+1}, \sigma_t) \gets \KeyUpd(\pk_t, \sk_t)$:
        Parse $\pk_t = (p, q, g, h)$.
        Set $s = \sk_t$.
        Randomly choose $s' \in \Z_q$.
        Set $d = s' - s \bmod p$ and $h' = hg^d \bmod p$.
        Output $(\pk_{t+1}, \sk_{t+1}, \sigma_t) = ((p, q, g, h'), s', (h, d))$.
        \item $\ct_{t+1} \gets \CtUpd(\ct_t, \sigma_t)$:
        Parse $\ct_t = (c_1, c_2)$ and $\sigma_t = (h, d)$.
        Randomly choose $r \in \Z_q$.
        Output $\ct_{t+1} = (c_1 g^r \bmod p, (c_1 g^r)^d c_2 h^r \bmod p)$.
    \end{itemize}
    The dynamic-key ElGamal encryption is updatable multiplicative homomorphic encryption, i.e., $\Dec(\sk_t, \Eval(\pk, \ct_{1,t}, \ct_{2,t})) =  \Dec(\sk_t, \Eval(\pk_t, \Enc(\pk_t, m_1), \Enc(\pk_t, m_2))) = m_1 m_2 \bmod p$.
\end{example}

This study quantifies the security level of an encryption scheme by the number of bits as follows~\cite{Katz21}.

\begin{definition}
\label{def:bit_sec}
    An encryption scheme satisfies $\lambda$ bit security if at least $2^\lambda$ operations are required for breaking the scheme.
\end{definition}

A security parameter in \defref{def:he} quantifies the level of bit security for (updatable) homomorphic encryption.
We address how to design the number of bits, $\lambda$, such that an encrypted control system becomes secure.

\subsection{Encrypted control}
\label{sec:ec}

This section introduces the syntax and security definition of encrypted control with updatable homomorphic encryption.

\begin{definition}
\label{def:ec}
    Given updatable homomorphic encryption and a controller $f:(\Phi, \xi) \mapsto \psi$, where $\Phi \in \X^{\alpha \times \beta}$ is a controller parameter, $\xi \in \X^\beta$ is a controller input, and $\psi \in \X^\alpha$ is a controller output.
    Suppose there exist an encoder $\Ecd$ and a decoder $\Dcd$ such that:
    \begin{itemize}
        \item $m \gets \Ecd(x; \Delta)$: An encoder algorithm takes $x \in \X$ and a scaling factor $\Delta \in \R$ as input and outputs a plaintext $m \in \M$.
        \item $x \gets \Dcd(m; \Delta)$: A decoder algorithm takes a plaintext $m \in \M$ and a scaling factor $\Delta \in \R$ as input and outputs $x \in \X$.
    \end{itemize}
    An encrypted controller of $f$ is $\EC$ such that:
    \begin{itemize}
        \item $\ct_\psi \gets \EC(\pk, \ct_\Phi, \ct_\xi)$: An encrypted control algorithm takes a public key $\pk$ and ciphertexts $\ct_\Phi \in \C^{\alpha \times \beta}, \ct_\xi \in \C^\beta$ as input and outputs a ciphertext $\ct_\psi \in \C^\alpha$.
        \item $\Dcd(\Dec(\sk_t, \EC(\pk_t, \ct_{\Phi,t}, \ct_{\xi,t})); \Delta) \simeq f(\Phi, \xi_t)$ holds for some $\Delta \in \R$, for all $t\in\Z^+$, for any $(\pk_0, \sk_0) \gets \KeyGen(1^\lambda)$, for any $\Phi \in \X^{\alpha \times \beta}$, and for any $\xi_t \in \X^\beta$, where $(\pk_{t+1}, \sk_{t+1}, \sigma_t) \gets \KeyUpd(\pk_t, \sk_t)$, $\ct_{\Phi,0} \gets \Enc(\pk_0, \Ecd(\Phi; \Delta))$, $\ct_{\Phi,t+1} \gets \CtUpd(\ct_{\Phi,t}, \sigma_t)$, $\ct_{\xi,t} \gets \Enc(\pk_t, \Ecd(\xi_t; \Delta))$, and the algorithms perform each element of matrices and vectors.
    \end{itemize}
\end{definition}

The controller parameter and input need to be encoded to plaintexts by the encoder $\Ecd$ before encryption because control systems typically operate over real numbers.
Although the encoding causes quantization errors, we ignore the errors for simplicity.

The security of encrypted control systems is defined based on a kind of sample complexities of system identification and computation time for breaking ciphertexts used in the system identification~\cite{teranishi2022a}.
The complexity and computatin time are called a sample identifying complexity and a sample deciphering time, respectively, defined as follows.

\begin{definition}
\label{def:sic}
    Let $N$ be a sample size for system identification by an adversary.
    A sample identifying complexity $\gamma$ is a function satisfying $\gamma(N) \le \EV[\epsilon(N)]$, where $\epsilon$ is an estimation error of the system identification.
\end{definition}

\begin{definition}
\label{def:sdt}
    Suppose an adversary uses a computer of $\Upsilon$~FLOPS.
    A sample deciphering time $\tau$ is a computation time required for breaking $N$ ciphertexts of an updatable homomorphic encryption that satisfies $\lambda$ bit security used for system identification by an adversary, namely $\tau(N, \lambda) = 2^\lambda N / \Upsilon$.
\end{definition}

The security of encrypted control systems is defined using the sample identifying complexity and sample deciphering time as follows.

\begin{definition}
\label{def:security}
    Let $\gamma_c$ be an acceptable estimation error, and $\tau_c$ be a defense period.
    An encrypted control system is secure if there does not exist a sample size $N$ such that $\gamma(N) < \gamma_c$ and $\tau(N, \lambda) \le \tau_c$, where $\gamma$ and $\tau$ are defined in \defref{def:sic} and \defref{def:sdt}, respectively.
    Otherwise, the encrypted control system is unsecure.
\end{definition}

Note that a pair of $\gamma_c$ and $\tau_c$ shows a security level of encrypted control systems and is used as design parameters for a security parameter later.

\begin{remark}
    The sample deciphering time in the case of using a typical homomorphic encryption with a fixed key pair is computed as $\tau(1, \lambda)$ regardless of a sample size $N$ because an adversary can obtain the original message of any ciphertext once the encryption scheme is broken.
    However,  the sample deciphering time in \defref{def:sdt} depends on $N$ because ciphertexts at different times are corresponding to different key pairs when updatable homomorphic encryption is used.
\end{remark}

\section{Threat Model}
\label{sec:threat}

This section formulates a threat model considered in this study.
\figref{fig:adversary} shows two types of adversaries that aim to identify system parameters.
Eve in \figref{fig:adversary}\subref{fig:sia_eavesdropper} is an adversary eavesdropping on network signals and exploiting illegal input signals to a communication channel from the encrypted controller to the decryptor.
This type of adversary represents man-in-the-middle attacks.
\figref{fig:adversary}\subref{fig:sia_server} depicts another adversary performing system identification.
In the figure, Eve is in a server that computes an encrypted control algorithm.
The adversary records inputs and outputs of the encrypted control algorithm and returns falsified outputs.
Thus, it is called a malicious server that represents a server infected by malware or spoofing as an authorized agent.
It should be noted here that the signal flow of encrypted control systems under the adversaries in \figref{fig:adversary} is the same structure.
Hence, we can deal with the attacks by a unified threat model without assuming the adversary types.

\begin{figure}[t]
    \centering
    \subfloat[Eavesdropper.]{\includegraphics[scale=1]{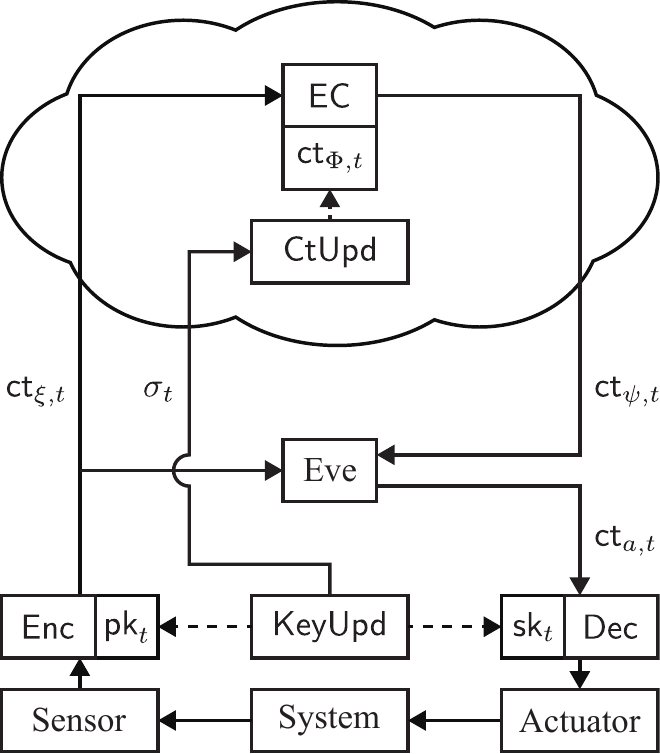}\label{fig:sia_eavesdropper}} \quad\quad
    \subfloat[Malicious server.]{\includegraphics[scale=1]{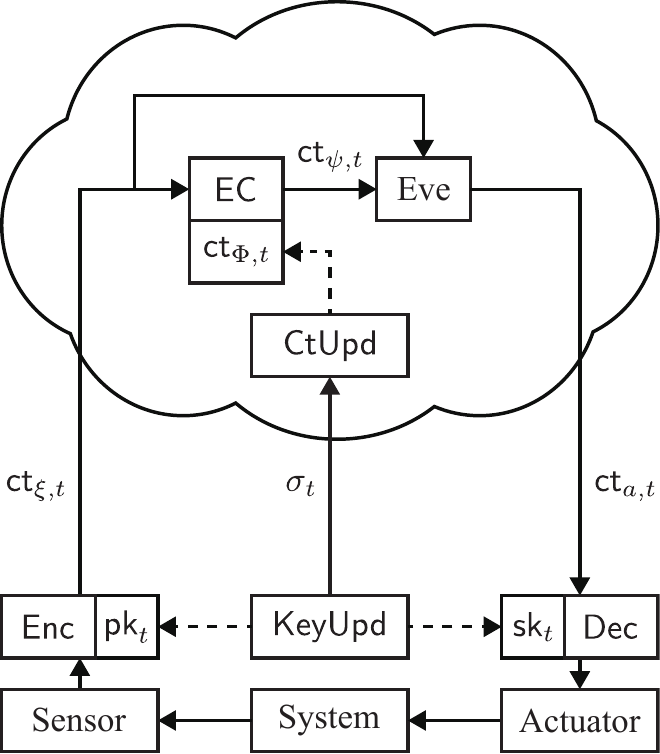}\label{fig:sia_server}}
    \caption{Two types of adversaries identifying the system.}
    \label{fig:adversary}
\end{figure}

Suppose the system in \figref{fig:adversary} is given as
\begin{equation}
    x_{t+1} = A x_t + B u_t + w_t,
    \label{eq:system}
\end{equation}
where $t \in \Z^+$ is a time, $x \in \R^n$ is a state, $u \in \R^m$ is an input, and $w \in \R^n$ is a noise.
Suppose $x_0$ and $w_t$ are independent and identically distributed over the Gaussian distribution with mean $\bm{0}$ and variance $\sigma_w^2 I$.
$A \in \R^{n \times n}$ and $B \in \R^{n \times m}$ are system parameters, and $A$ is assumed to be stable.
The state of \eqref{eq:system} is encrypted by updatable homomorphic encryption as $\ct_{x,t} \gets \Enc(\pk_t, \Ecd(x_t; \Delta))$ and transmitted to a controller server, where $(\pk_0, \sk_0) \gets \KeyGen(1^\lambda)$, and $(\pk_{t+1}, \sk_{t+1}, \sigma_t) \gets \KeyUpd(\pk_t, \sk_t)$.
The server returns an input ciphertext $\ct_{u,t} \gets \EC(\pk_t, \ct_{\Phi,t}, \ct_{x,t})$ to the system, where $\Phi$ is a controller parameter, $\ct_{\Phi,0} \gets \Enc(\pk, \Ecd(\Phi; \Delta))$, and $\ct_{\Phi,t+1} \gets \CtUpd(\ct_{\Phi,t}, \sigma_t)$.
The system decrypts the input ciphertext and obtains an input as $u_t \gets \Dcd(\Dec(\sk_t, \ct_{u,t}); \Delta)$.

This study considers an adversary following the protocol: 1) collecting some encrypted samples, 2) exposing the original data by breaking the samples, and 3) identifying system parameters $(A, B)$ by a least squares method with the exposed data.
The attack scenario is formally defined as follows.

\begin{definition}
\label{def:adversary}
    The adversary attempts to identify $(A, B)$ of \eqref{eq:system} by the following procedure.
    \begin{enumerate}
        \item The adversary injects malicious inputs $u_t = a_t$ for $t \in [t_s, t_f]$ and collects $N = t_f - t_s + 1$ pairs of input and state ciphertexts $\{(\ct_{u,t}, \ct_{x,t})\}_{t=t_s}^{t_f}$.
        \item The adversary exposes $\{(u_t, x_t)\}_{t=t_s}^{t_f}$ deciphering the ciphertexts.
        \item The adversary estimates $(A, B)$ by a least squares method with the exposed data.
    \end{enumerate}
\end{definition}

For the third step in \defref{def:adversary}, we consider the following least squares identification method.
Define data matrices
\begin{alignat*}{2}
    X_f &=
    \begin{bmatrix}
        x_{t_s+1} & \cdots & x_{t_f}
    \end{bmatrix}, &\quad
    X_{p} &=
    \begin{bmatrix}
        x_{t_s} & \cdots & x_{t_f-1}
    \end{bmatrix}, \\
    U_p &=
    \begin{bmatrix}
        u_{t_s} & \cdots & u_{t_f-1}
    \end{bmatrix}, &\quad
    W_p &=
    \begin{bmatrix}
        w_{t_s} & \cdots & w_{t_f-1}
    \end{bmatrix}.
\end{alignat*}
It follows from \eqref{eq:system} that
\begin{equation}
    X_f = A X_p + B U_p + W_p =
    \begin{bmatrix}
        A & B
    \end{bmatrix}
    \begin{bmatrix}
        X_p \\
        U_p
    \end{bmatrix} + W_p.
    \label{eq:data}
\end{equation}
The least squares estimators $(\hat{A},\hat{B})$ of $(A,B)$ are given as
\begin{equation}
    \begin{bmatrix}
        \hat{A} & \hat{B}
    \end{bmatrix} =
    \argmin_{[A \ B]} \left\| X_f -
    \begin{bmatrix}
        A & B
    \end{bmatrix}
    \begin{bmatrix}
        X_p \\
        U_p
    \end{bmatrix}
    \right\|_F^2 = X_f
    \begin{bmatrix}
        X_p \\
        U_p
    \end{bmatrix}^+,
    \label{eq:estimate}
\end{equation}
where $([X_p^\top\ U_p^\top]^\top)^+$ is the pseudo inverse matrix of $[X_p^\top\ U_p^\top]^\top$.

\begin{remark}
    In the first step of \defref{def:adversary}, the malicious inputs $a_t$ can be injected properly even though control inputs are encrypted by updatable homomorphic encryption because, in general, an encryption scheme and a public key are public information.
    Furthermore, even if an adversary does not know a public key, the adversary can falsify ciphertexts using malleability~\cite{cheon2018a,teranishi2019a,cheon2020,fauser2020}.
\end{remark}

\section{Secure Updatable Homomorphic Encryption Against Malicious Server}
\label{sec:updatable_homomorphic_encryption}

This section presents a modification of the updatable homomorphic encryption scheme in \exref{ex:dyn_elgamal}.
To begin with, we introduce a desired cryptographic property of the encryption scheme~\cite{teranishi2022a}.

\begin{proposition}
\label{prop:impossibility}
    Consider the updatable homomorphic encryption in \exref{ex:dyn_elgamal}.
    Suppose an adversary has $\pk_t$, $\sk_t$, and $\ct_t$.
    The probabilities $\Pr(\hat{\sk}_{t-1} = \sk_{t-1})$ and $\Pr(\hat{\sk}_{t+1} = \sk_{t+1})$ are negligibly small for all $t \in \N$, for any $(\pk_0, \sk_0)$, and for any $\ct_0$, where $(\pk_{t+1}, \sk_{t+1}, \sigma_t) \gets \KeyUpd(\pk_t, \sk_t)$, $\ct_{t+1} \gets \CtUpd(\ct_t, \sigma_t)$, and $\hat{\sk}_{t-1}$ and $\hat{\sk}_{t+1}$ are adversary's estimates of $\sk_{t-1}$ and $\sk_{t+1}$, respectively.
\end{proposition}

\begin{proof}
    See Proposition 2 in~\cite{teranishi2022a}.
\end{proof}

The proposition implies the impossibility for estimating the previous and next secret keys from the current secret key.
Hence, the proposition is the foundation for that the sample deciphering time in \defref{def:sdt} depends on a sample size $N$ because an adversary must keep breaking $N - 1$ ciphertexts even though the adversary succeeds to break one of $N$ ciphertexts.
However, the impossibility makes sense only for a network eavesdropper because the proposition is satisfied as long as an update token is secret against the adversary.
The following proposition reveals that there exists a simple attack to obtain the next secret key from the current secret key and update token.

\begin{proposition}
\label{prop:vulnerability}
    Consider the updatable homomorphic encryption in \exref{ex:dyn_elgamal}.
    Suppose an adversary has $\sk_t$ and $\sigma_t$.
    Then, the adversary can achieve $\Pr(\hat{\sk}_{t+1} = \sk_{t+1}) = 1$ for all $t \in \N$ and for any $(\pk_0, \sk_0) \gets \KeyGen(1^\lambda)$, where $(\pk_{t+1}, \sk_{t+1}, \sigma_t) \gets \KeyUpd(\pk_t, \sk_t)$, and $\hat{\sk}_{t+1}$ is adversary's estimate of $\sk_{t+1}$.
\end{proposition}

\begin{proof}
    Let $\sk_t = s$ and $\sk_{t+1} = s'$.
    Here $d = s' - s$ and $\sigma_t = (h, d)$ for some $h$, and thus the adversary can estimate $\sk_{t+1}$ as $\hat{\sk}_{t+1} = \sk_t + d = s + (s' - s) = s'$.
\end{proof}

By the proposition, the conventional encryption scheme cannot satisfy the impossibility against a malicious server who must has an update token for updating a controller parameter ciphertext as in \defref{def:ec}.
This study presents the modified homomorphic evaluation and decryption algorithms to solve this problem.

\begin{definition}
\label{def:modification}
    Consider the encryption scheme in \exref{ex:dyn_elgamal}.
    Define a modified homomorphic evaluation algorithm $\overline{\Eval}$ and a modified decryption algorithm $\overline{\Dec}$ as follows.
    \begin{itemize}
        \item $\overline{\ct} \gets \overline{\Eval}(\pk, \ct_1, \ct_2)$:
        Compute $\ct \gets \Eval(\pk, \ct_1, \ct_2)$.
        Parse $\ct_1 = (c_{11}, c_{12})$, $\ct_2 = (c_{21}, c_{22})$, and $\ct = (c_1, c_2)$.
        Return $\overline{\ct} = (c_{11}, c_{21}, c_2)$.
        \item $m \gets \overline{\Dec}(\sk_1, \sk_2, \overline{\ct})$:
        Parse $\overline{\ct} = (c_1, c_2, c_3)$.
        Compute $\tilde{c} \gets \Dec(\sk_2, (c_2, c_3))$.
        Return $m \gets \Dec(\sk_1, (c_1, \tilde{c}))$.
    \end{itemize}
\end{definition}

The homomorphism of original homomorphic evaluation algorithm in \exref{ex:dyn_elgamal} holds only for two ciphertexts of the same time.
In contrast, the modified algorithm can satisfy the homomorphism with two ciphertexts of different times.

\begin{theorem}
    Let $k \in \N$.
    The encryption scheme in \exref{ex:dyn_elgamal} with the modified algorithms in \defref{def:modification} satisfies
    \[
        \overline{\Dec}(\sk_t, \sk_{t+k}, \overline{\Eval}(\pk_t, \Enc(\pk_t, m_1), \Enc(\pk_{t+k}, m_2))) = m_1 m_2 \bmod p
    \]
    for any $(\pk_0, \sk_0) \gets \KeyGen(1^\lambda)$, for any $m_1, m_2 \in \M$, and for all $t \in \Z^+$, where $(\pk_{t+1}, \sk_{t+1}, \sigma_t) \gets \KeyUpd(\pk_t, \sk_t)$.
\end{theorem}

\begin{proof}
    Let $\sk_t = s$, $\sk_{t+k} = s'$, $\pk_t = (p, q, g, g^s \bmod p)$, and $\pk_{t+k} = (p, q, g, g^{s'} \bmod p)$.
    Then, $\overline{\Eval}(\pk_t, \Enc(\pk_t, m_1), \Enc(\pk_{t+k}, m_2)) = (c_1, c_2, c_3) = (g^r \bmod p, g^{r'} \bmod p, m_1 m_2 g^{sr+s'r'} \bmod p)$, where $r$ and $r'$ are random numbers corresponding to times $t$ and $t + k$, respectively.
    The intermediate output $\tilde{c}$ is obtained as $\tilde{c} = \Dec(\sk_{t+k}, (c_2, c_3)) = g^{-s'r'} m_1 m_2 g^{sr+s'r'} = m_1 m_2 g^{sr} \bmod p$.
    Therefore, $\overline{\Dec}(\sk_t, \sk_{t+k}, (c_1, c_2, c_3)) = \Dec(\sk_t, (c_1, \tilde{c})) = g^{-sr} m_1 m_2 g^{sr} = m_1 m_2 \bmod p$.
\end{proof}

With the algorithms in \exref{ex:dyn_elgamal} and \defref{def:modification}, the encrypted control algorithm in \defref{def:ec} of a linear controller $(\Phi, \xi_t) \mapsto \psi_t = \Phi \xi_t$ can be implemented as
\begin{equation}
    \EC(\pk_0, \ct_{\Phi,0}, \ct_{\xi,t}) =
    \begin{bmatrix}
        \overline{\Eval}(\pk_0, \ct_{\Phi_{11},0}, \ct_{\xi_1,t})      & \cdots & \overline{\Eval}(\pk_0, \ct_{\Phi_{1\beta},0}, \ct_{\xi_\beta,t}) \\
        \vdots                                                         & \ddots & \vdots                                                            \\
        \overline{\Eval}(\pk_0, \ct_{\Phi_{\alpha1},0}, \ct_{\xi_1,t}) & \cdots & \overline{\Eval}(\pk_0, \ct_{\Phi_{\alpha\beta},0}, \ct_{\xi_\beta,t})
    \end{bmatrix},
    \label{eq:ec}
\end{equation}
where the decryption algorithm in \defref{def:ec} is given as $\Sum \circ \overline{\Dec}$, and $\Sum: \M^{m \times n} \to \M^m : M \mapsto [\, \sum_{i=1}^n M_{1i} \ \allowbreak \cdots \ \allowbreak \sum_{i=1}^n M_{mi} \,]^\top$~\cite{kogiso2015}.
\figref{fig:ecs} shows the encrypted control system using the modified updatable homomorphic encryption that operates without transmitting an update token $\sigma_t$ from the system to the controller server.
Note that an encoder $\Ecd$ and a decoder $\Dcd$ are omitted in the figure for simplicity.
The controller server receives $\ct_{\xi,t} \gets \Enc(\pk_t, \Ecd(\xi_t; \Delta))$ at every time and returns $\overline{\ct}_{\psi,t} \gets \EC(\pk_0, \ct_{\Phi,0}, \ct_{\xi,t})$, where $\ct_{\Phi,0} \gets \Enc(\pk_0, \Ecd(\Phi; \Delta))$, while public and secret keys are updated by $(\pk_{t+1}, \sk_{t+1}, \sigma_t) \gets \KeyUpd(\pk_t, \sk_t)$.
The system recovers a controller output as $\psi_t \gets \Dcd(\Sum(\overline{\Dec}(\sk_0, \sk_{t}, \overline{\ct}_{\psi,t})); \Delta)$.
Consequently, the modification in \defref{def:modification} is beneficial for achieving the impossibility against not only an eavesdropper but also a malicious server.

\begin{figure}[t]
    \centering
    \includegraphics[scale=1]{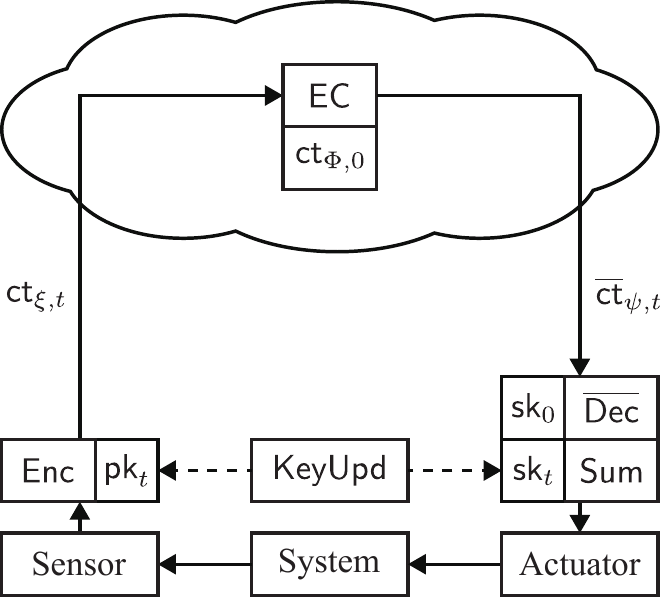}
    \caption{Encrypted control system with the modified updatable homomorphic encryption.}
    \label{fig:ecs}
\end{figure}

\section{Security Parameter Design}
\label{sec:security_parameter_design}

This section proposes a design method for a security parameter of the modified updatable homomorphic encryption that consists of the algorithms in \exref{ex:dyn_elgamal} and \defref{def:modification}.
To this end, we propose a novel sample identifying complexity of \eqref{eq:system} with the encrypted controller \eqref{eq:ec} under the adversary in \defref{def:adversary}.
Using the sample identifying complexity, we design the minimum security parameter that makes the encrypted control system secure against the adversary.

A sample identifying complexity and a sample deciphering time are crucial for defining the security of encrypted control systems in \defref{def:security}.
The sample deciphering time in \defref{def:sdt} can be computed without assuming a used encryption scheme.
In contrast, a computation method for a sample identifying complexity is not obvious because it depends on system dynamics and a system identification method.
This study proposes a sample identifying complexity of \eqref{eq:system} under the adversary in \defref{def:adversary} when the estimation error of least squares identification method is defined as follows.

\begin{definition}
\label{def:error}
    The estimation error $\epsilon$ of \eqref{eq:estimate} is defined as
    \[
        \epsilon(N) = \cfrac{1}{c} \left\|
        \begin{bmatrix}
            A & B
        \end{bmatrix} -
        \begin{bmatrix}
            \hat{A} & \hat{B}
        \end{bmatrix}
        \right\|_F^2,
    \]
    where $c=n(n+m)$ is the number of entries of $A$ and $B$.
\end{definition}

By \defref{def:error}, $\epsilon$ is a mean square error of the estimates $\hat{A}$ and $\hat{B}$.
It should be noted here that one of the best strategies for the adversary in \defref{def:adversary} to design the malicious inputs $a_{t_s}, \dots, a_{t_f}$ minimizing the error $\epsilon$ is that the inputs are independently and identically sampled from the Gaussian distribution with mean zero.
Under this setting, the following theorem reveals a sample identifying complexity.

\begin{theorem}
\label{thm:sic}
    Suppose malicious inputs $a_{t_s}, \dots, a_{t_f}$ are i.i.d. signals following the Gaussian distribution with mean $\bm{0}$ and variance $\sigma_u^2I$.
    The function
    \begin{equation}
        \gamma(N) = \cfrac{(m + n) \sigma_w^2}{\displaystyle \sigma_x^2 \tr(\Psi_w) + (N - 1) \left[ \sigma_u^2 (\tr(\Psi_u) + m) + \sigma_w^2 \tr(\Psi_w) \right]}
        \label{eq:sic}
    \end{equation}
    is the sample identifying complexity of \eqref{eq:system} under the adversary in \defref{def:adversary}, where $\Psi_u$ and $\Psi_w$ are controllability Gramians obtained by solving the discrete Lyapunov equations, $A \Psi_u A^\top - \Psi_u + B B^\top = 0$ and $A \Psi_w A^\top - \Psi_w + I = 0$, respectively.
\end{theorem}

\begin{proof}
    Let $D = [X_p^\top \ U_p^\top]^\top$.
    It follows from \eqref{eq:data} and \eqref{eq:estimate} that
    \begin{align*}
            \EV[\epsilon(N)]
        &=  \cfrac{1}{c} \EV\left[ \left\|
            \begin{bmatrix}
                A & B
            \end{bmatrix}
            - X_f D^+ \right\|_F^2 \right], \\
        &=  \cfrac{1}{c} \EV\left[ \left\|
            \begin{bmatrix}
                A & B
            \end{bmatrix}
            - \left(
            \begin{bmatrix}
                A & B
            \end{bmatrix}
            D + W_p \right) D^+ \right\|_F^2 \right], \\
        &=  \cfrac{1}{c} \EV\left[ \left\| W_p D^+ \right\|_F^2 \right], \\
        &=  \cfrac{1}{c} \EV\left[ \left\| \vec(W_p D^+) \right\|_2^2 \right], \\
        &=  \cfrac{1}{c} \EV\left[ \tr\left( \vec(W_p D^+) \vec(W_p D^+)^\top \right) \right], \\
        &=  \cfrac{1}{c} \EV\left[ \tr\left( ( D^+ \otimes I )^\top \vec(W_p)  \vec(W_p)^\top ( D^+ \otimes I ) \right) \right], \\
        &=  \cfrac{1}{c} \EV\left[ \tr\left( \left( D^+ (D^+)^\top \otimes I \right) \vec(W_p) \vec(W_p)^\top \right) \right], \\
        &=  \cfrac{1}{c} \tr\left( \EV\left[ D^+ (D^+)^\top \right] \EV\left[
            \begin{bmatrix}
                w_{t_s}^\top w_{t_s} &        & \\
                                     & \ddots & \\
                                     &        & w_{t_f-1}^\top w_{t_f-1}
            \end{bmatrix}
            \right] \right), \\
        &=  \cfrac{\sigma_w^2}{m + n} \tr\left( \EV\left[ D^\top ( D D^\top )^{-1} \left( D^\top ( D D^\top )^{-1} \right)^\top \right] \right), \\
        &=  \cfrac{\sigma_w^2}{m + n} \tr\left( \EV\left[ ( D D^\top )^{-1} \right] \right),
    \end{align*}
    where $\otimes$ is the Kronecker product.
    Using Jensen's inequality, the expectation of trace of inverse matrix is bounded from below by
    \begin{align*}
                \tr\left( \EV\left[ (D D^\top)^{-1} \right] \right) 
        &\ge    (m + n)^2 \EV\left[ \tr\left( D D^\top \right)^{-1} \right], \\
        &\ge    (m + n)^2 \EV\left[ \tr\left( D D^\top \right) \right]^{-1}, \\
        &=      (m + n)^2 \EV\left[ \tr\left(
                \begin{bmatrix}
                    X_p \\
                    U_p
                \end{bmatrix}
                \begin{bmatrix}
                    X_p^\top & U_p^\top
                \end{bmatrix}
                \right) \right]^{-1}, \\
        &=      (m + n)^2 \EV\left[ \tr\left(
                \begin{bmatrix}
                    X_p X_p^\top & X_p U_p^\top \\
                    U_p X_p^\top & U_p U_p^\top
                \end{bmatrix}
                \right) \right]^{-1}, \\
        &=      (m + n)^2 \left( \EV\left[ \tr\left( X_p X_p^\top \right) \right] + \EV\left[ \tr\left( U_p U_p^\top \right) \right] \right)^{-1}, \\
        &=      (m + n)^2 \left( \EV\left[ \tr\left( \sum_{t=t_s}^{t_f-1} x_t x_t^\top \right) \right] + \EV\left[ \tr\left( \sum_{t=t_s}^{t_f-1} u_t u_t^\top \right) \right] \right)^{-1}.
    \end{align*}
    It follows from \eqref{eq:system} that
    \[
        x_t = A^t x_0 + \sum_{k=0}^{t-1} A^{t-1-k} B u_k + \sum_{k=0}^{t-1} A^{t-1-k} w_k.
    \]
    Thus, the expectations of traces are given as
    \begin{align*}
                \EV\left[ \tr\left( \sum_{t=t_s}^{t_f-1} x_t x_t^\top \right) \right]
        &=      \EV\left[ \tr\left( \sum_{t=t_s}^{t_f-1} A^t x_0 x_0^\top (A^t)^\top \right) \right] \\
        &\quad+ \EV\left[ \tr\left( \sum_{t=t_s}^{t_f-1} \sum_{k=0}^{t-1} A^{t-1-k} B u_k u_k^\top B^\top (A^{t-1-k})^\top \right) \right] \\
        &\quad+ \EV\left[ \tr\left( \sum_{t=t_s}^{t_f-1} \sum_{k=0}^{t-1} A^{t-1-k} w_k w_k^\top (A^{t-1-k})^\top \right) \right], \\
        &=      \sigma_x^2 \tr\left( \sum_{t=t_s}^{t_f-1} A^t (A^t)^\top \right) + \sigma_u^2 \tr\left( \sum_{t=t_s}^{t_f-1} \sum_{k=0}^{t-1} A^k B B^\top (A^k)^\top \right) \\
        &\quad+ \sigma_w^2 \tr\left( \sum_{t=t_s}^{t_f-1} \sum_{k=0}^{t-1} A^k (A^k)^\top \right)
    \end{align*}
    and
    \[
        \EV\left[ \tr\left( \sum_{t=t_s}^{t_f-1} u_t u_t^\top \right) \right] = (N - 1) m \sigma_u^2.
    \]
    Furthermore, the matrices are bounded by
    \begin{align*}
        &\sum_{t=t_s}^{t_f-1} A^t (A^t)^\top \le \sum_{t=0}^\infty A^t (A^t)^\top = \Psi_w, \\
        &\sum_{k=0}^{t-1} A^k (A^k)^\top \le \sum_{k=0}^\infty A^k (A^k)^\top = \Psi_w, \\
        &\sum_{k=0}^{t-1} A^k B B^\top (A^k)^\top \le \sum_{k=0}^\infty A^k B B^\top (A^k)^\top = \Psi_u.
    \end{align*}
    Therefore, we obtain
    \begin{align*}
                \EV[\epsilon(N)]
        &\ge    \cfrac{\sigma_w^2}{m + n} \cdot \cfrac{(m + n)^2}{\sigma_x^2 \tr(\Psi_w) + (N - 1) \sigma_u^2 \tr(\Psi_u) + (N - 1) \sigma_w^2 \tr(\Psi_w) + (N - 1) m \sigma_u^2}, \\
        &=      \cfrac{(m + n) \sigma_w^2}{\displaystyle \sigma_x^2 \tr(\Psi_w) + (N - 1) \left[ \sigma_u^2 (\tr(\Psi_u) + m) + \sigma_w^2 \tr(\Psi_w) \right]} = \gamma(N).
    \end{align*}
    By \defref{def:sic}, $\gamma(N)$ is the sample identifying complexity of \eqref{eq:system} under the adversary in \defref{def:adversary}
\end{proof}

If a sample size is sufficiently large, the sample identifying complexity \eqref{eq:sic} is given as a simple equation.

\begin{corollary}
    Let $R_\sigma = \sigma_u^2 / \sigma_w^2$.
    Suppose a sample size $N$ is sufficiently large.
    Then, the function
    \begin{equation}
        \gamma(N) = \cfrac{m + n}{\displaystyle (N - 1) \left[ R_\sigma (\tr(\Psi_u) + m) + \tr(\Psi_w) \right]}
        \label{eq:sic2}
    \end{equation}
    is the sample identifying complexity of \eqref{eq:system} under the adversary in \defref{def:adversary}.
\end{corollary}

\begin{proof}
    If $N$ is sufficiently large, the denominator of \eqref{eq:sic} can be approximated by $(N - 1) \left[ \sigma_u^2 (\tr(\Psi_u) + m) + \sigma_w^2 \tr(\Psi_w) \right]$.
    Then, \eqref{eq:sic2} holds by dividing both the numerator and denominator of \eqref{eq:sic} by $\sigma_w^2$.
\end{proof}

The equation \eqref{eq:sic2} shows that the sample identifying complexity is characterized by the traces of controllability Gramians $\Psi_u, \Psi_w$ and variance ratio $R_\sigma$.
If $R_\sigma$ is small, i.e., $\sigma_u^2 \ll \sigma_w^2$, the sample identifying complexity can be approximated by
\[
    \gamma(N) \simeq \cfrac{m + n}{(N - 1) \tr(\Psi_w)},
\]
and system states are driven by almost only system noises.
In such a case, the smaller eigenvalues of $\Psi_w$ that represent the degree of effects from the noises to the states are, the larger sample identifying complexity is.
In contrast, if $R_\sigma$ is large, i.e., $\sigma_u^2 \gg \sigma_w^2$, the sample identifying complexity can be approximated by
\[
    \gamma(N) \simeq \cfrac{m + n}{(N - 1) R_\sigma (\tr(\Psi_u) + m)},
\]
and the states are driven by almost only system inputs rather than the noises.
The sample identifying complexity in this case increases as the trace of $\Psi_u$ decreases.

The observations suggest a defense policy that minimizes the eigenvalues of Gramians to reduce the information leakage of \eqref{eq:system} by maximizing the sample identifying complexity.
However, the defense policy seems to have a limitation.
An adversary may choose an input variance $\sigma_u^2$ sufficiently larger than a noise variance $\sigma_w^2$ for decreasing the estimation error.
Then, the sample identifying complexity converges to
\begin{equation}
    \bar{\gamma}(N) = \cfrac{m + n}{(N - 1) m R_\sigma}
    \label{eq:bound}
\end{equation}
as the trace of Gramian $\Psi_u$ goes to zero.
The equation \eqref{eq:bound} is the upperbound of sample identifying complexity when $R_\sigma$ is large.
Furthermore, reducing the trace of $\Psi_u$ implies that the energy of system inputs affecting system states is attenuated.
In other words, the controllability of \eqref{eq:system} should be worse for improving the sample identifying complexity.
This property is not desired in practice because it means that the system is difficult to control.
Note that, even when $\sigma_u^2$ is sufficiently smaller than $\sigma_w^2$, there is the upperbound
\[
    \bar{\gamma}(N) = \cfrac{m + n}{(N - 1) n}
\]
because $A \Psi_w A^\top - \Psi_w + I = 0$ holds only if $\tr(\Psi_w) > \tr(I) = n$ as long as $A$ is not a zero matrix.

The upperbounds motivate to increase a security parameter of a used encryption scheme for further improving the security.
Meanwhile, a large security parameter leads to a high computational burden.
This dilemma can be solved reasonably by obtaining the optimal security parameter designed as the minimum security parameter that guarantees the security of encrypted control system.
The security parameter design in this study follows the approach in~\cite{teranishi2023} using the sample identifying complexity \eqref{eq:sic2}.
The rest of this section describes the summary of this approach.
The sample deciphering time in \defref{def:sdt} is monotonically increasing on a sample size $N$.
Hence, by \defref{def:security}, an encrypted control system becomes secure if the sample deciphering time $\tau(N^\ast, \lambda)$ becomes larger than a defense period $\tau_c$, where $N^\ast$ is the minimum sample size such that the sample identifying complexity $\gamma(N^\ast)$ is smaller than an acceptable estimation error $\gamma_c$.
Consequently, we obtain the following theorem.

\begin{theorem}
\label{thm:opt_sec}
    Suppose a sample size $N$ is sufficiently large.
    The minimum security parameter $\lambda^\ast$ guarantees that the encrypted control system consisting of \eqref{eq:system} and \eqref{eq:ec} becomes secure, in the sense of \defref{def:security}, is
    \begin{equation}
        \lambda^\ast = \left\lfloor \log_2 \cfrac{\Upsilon \tau_c}{N^\ast} \right\rfloor + 1, \quad N^\ast = \left\lfloor \cfrac{m + n}{\gamma_c \left[ R_\sigma (\tr(\Psi_u) + m) + \tr(\Psi_w) \right]} \right\rfloor + 2,
        \label{eq:opt_sec}
    \end{equation}
    where $\Upsilon$ and $(\gamma_c, \tau_c)$ are defined in \defref{def:sdt} and \defref{def:security}, respectively.
\end{theorem}

\begin{proof}
    It follows from \eqref{eq:sic2} that
    \[
        \gamma(N) < \gamma_c \iff N > \cfrac{m + n}{\gamma_c \left[ R_\sigma (\tr(\Psi_u) + m) + \tr(\Psi_w) \right]} + 1.
    \]
    Hence, the minimum sample size $N^\ast$ such that $\gamma(N^\ast) < \gamma_c$ is given as
    \[
        N^\ast = \left\lfloor \cfrac{m + n}{\gamma_c \left[ R_\sigma (\tr(\Psi_u) + m) + \tr(\Psi_w) \right]} + 1 \right\rfloor + 1.
    \]
    Similarly, the minimum security parameter $\lambda^\ast$ such that $\tau(N^\ast, \lambda^\ast) > \tau_c$ is given as
    \[
        \lambda^\ast = \left\lfloor \log_2 \cfrac{\Upsilon \tau_c}{N^\ast} \right\rfloor + 1,
    \]
    where
    \[
        \tau(N^\ast, \lambda) > \tau_c \iff \lambda > \log_2 \cfrac{\Upsilon \tau_c}{N^\ast}.
    \]
    This completes the proof.
\end{proof}

Note that the minimum key length $k^\ast$ of an encryption scheme that satisfies $\lambda^\ast$ bit security can be computed as
\begin{equation}
    k^\ast = \argmin_{k\in\N} \Omega(k) \quad \text{s.t.} \quad \Omega(k) \ge 2^{\lambda^\ast},
    \label{eq:opt_key}
\end{equation}
where $\lambda^\ast$ is given by \eqref{eq:opt_sec}, and $\Omega(k)$ is the time complexity of fastest known algorithm for breaking the encryption scheme.

\section{Numerical Simulation}
\label{sec:simulation}

This section presents the results of numerical simulations.
We set $m = n = 4$ and $\sigma_x^2 = 1$ throughout the simulations.

Consider the system \eqref{eq:system} whose controllability Gramians are $\Psi_w = \Psi_u = 2I$, where the corresponding system parameters are $A = 0.7071I$ and $B = I$.
\figref{fig:variance_ratio} shows the estimation errors and sample identifying complexities with the nine combinations of $\sigma_w^2=0.1, 1, 10$ and $\sigma_u^2=0.1, 1, 10$.
The gray dots are the estimation errors in \defref{def:error}.
The blue solid and orange dashed lines are the expectations of estimation errors and the sample identifying complexities \eqref{eq:sic2}, respectively.
Here, the system identification is performed $50$ times for each sample size with different data sets based on the dynamics of \eqref{eq:system} with the system parameters.
The estimation errors and their expectations in the figure are smaller as the variance ratio increases, and the proposed complexities capture the behavior of expectations in all the cases.
Moreover, the sample identifying complexity with the larger variance ratio is less conservative.
Hence, our proposed complexity becomes more practical as an adversary attempts to estimate system parameters more accurately.

\begin{figure}[t]
    \centering
    \subfloat[$R_\sigma = 1$, $\sigma_w^2 = 0.1$, $\sigma_u^2 = 0.1$.]{\includegraphics[scale=.6]{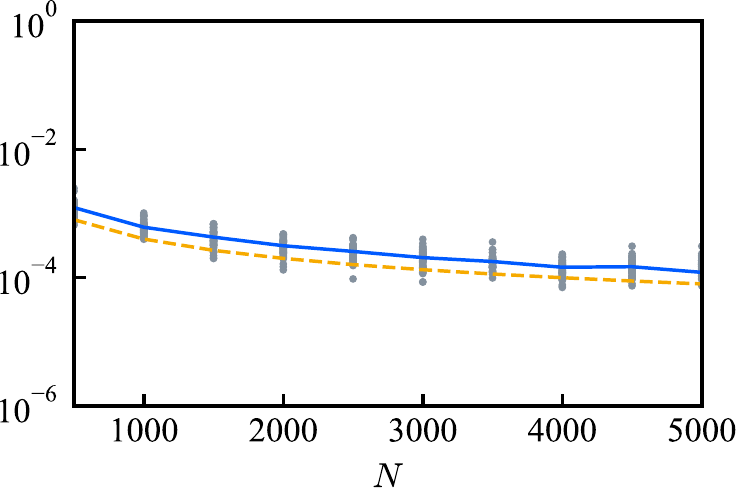}} \quad
    \subfloat[$R_\sigma = 10$, $\sigma_w^2 = 0.1$, $\sigma_u^2 = 1$.]{\includegraphics[scale=.6]{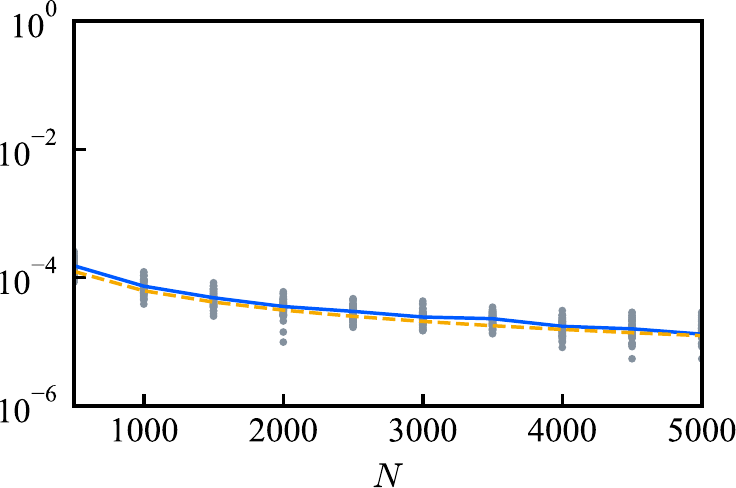}} \quad
    \subfloat[$R_\sigma = 100$, $\sigma_w^2 = 0.1$, $\sigma_u^2 = 10$.]{\includegraphics[scale=.6]{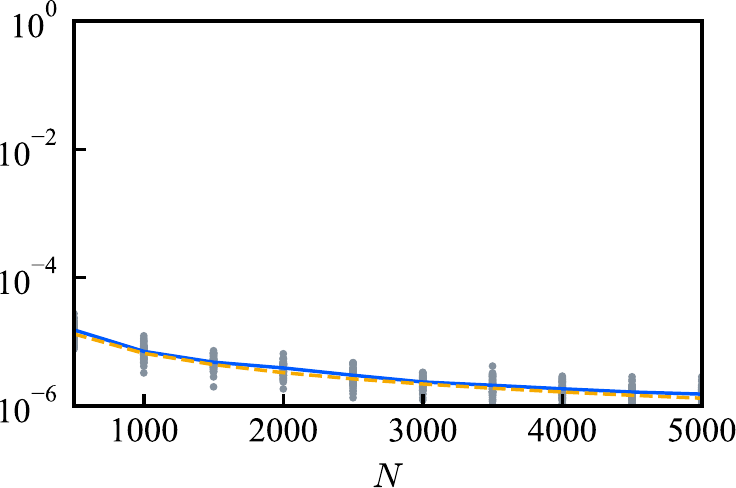}} \\
    \subfloat[$R_\sigma = 0.1$, $\sigma_w^2 = 1$, $\sigma_u^2 = 0.1$.]{\includegraphics[scale=.6]{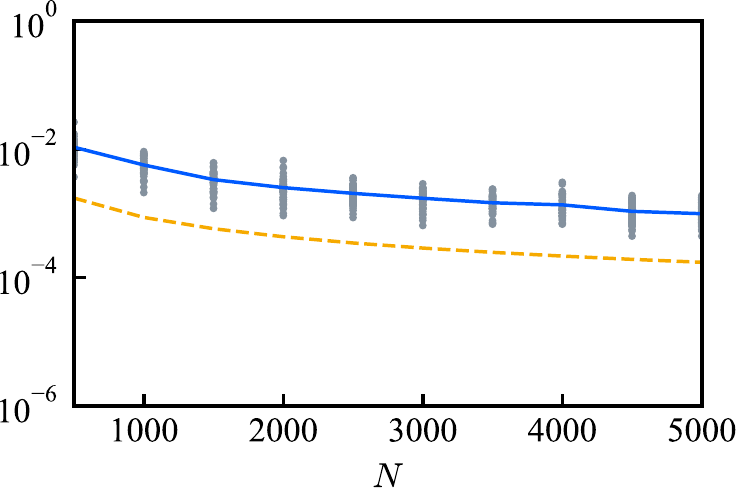}} \quad
    \subfloat[$R_\sigma = 1$, $\sigma_w^2 = 1$, $\sigma_u^2 = 1$.]{\includegraphics[scale=.6]{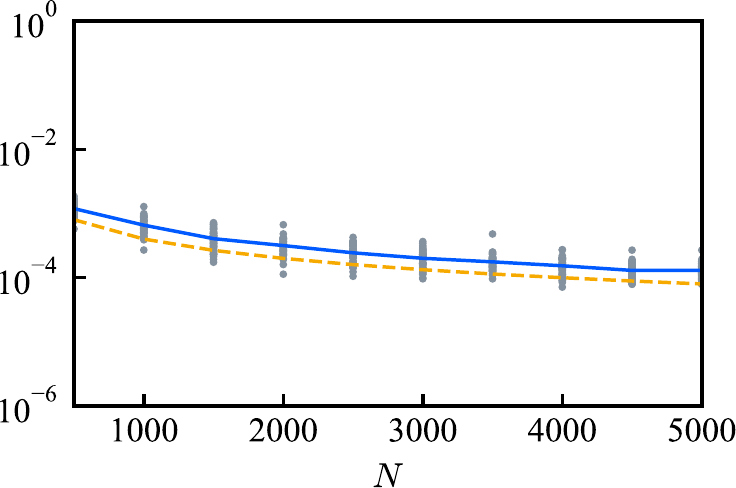}} \quad
    \subfloat[$R_\sigma = 10$, $\sigma_w^2 = 1$, $\sigma_u^2 = 10$.]{\includegraphics[scale=.6]{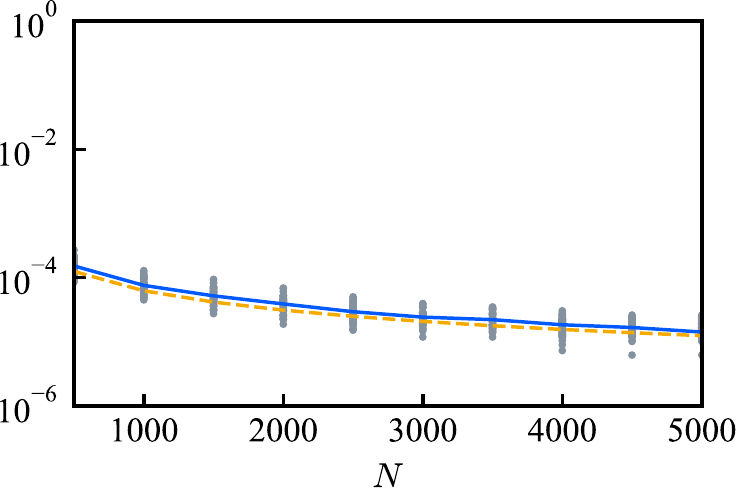}} \\
    \subfloat[$R_\sigma = 0.01$, $\sigma_w^2 = 10$, $\sigma_u^2 = 0.1$.]{\includegraphics[scale=.6]{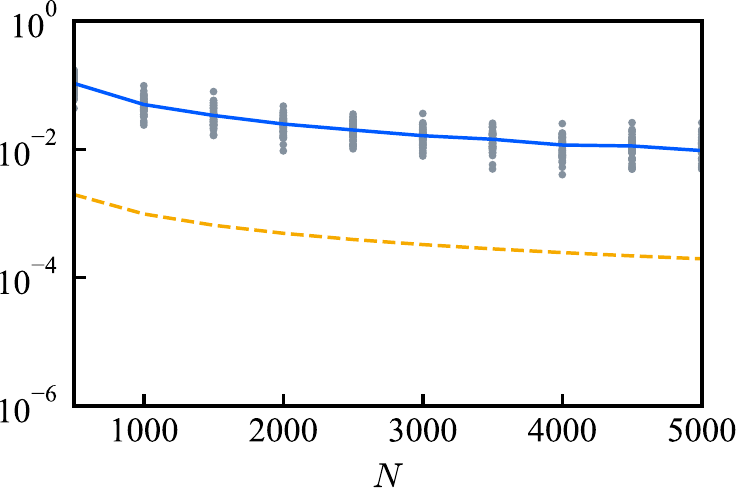}} \quad
    \subfloat[$R_\sigma = 0.1$, $\sigma_w^2 = 10$, $\sigma_u^2 = 1$.]{\includegraphics[scale=.6]{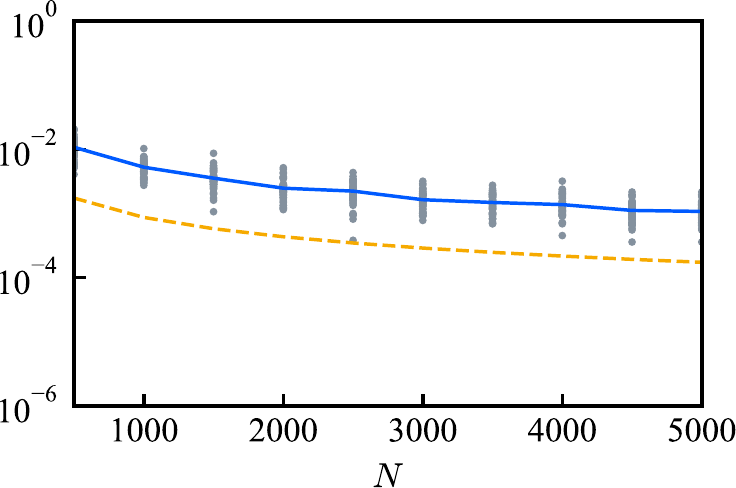}} \quad
    \subfloat[$R_\sigma = 1$, $\sigma_w^2 = 10$, $\sigma_u^2 = 10$.]{\includegraphics[scale=.6]{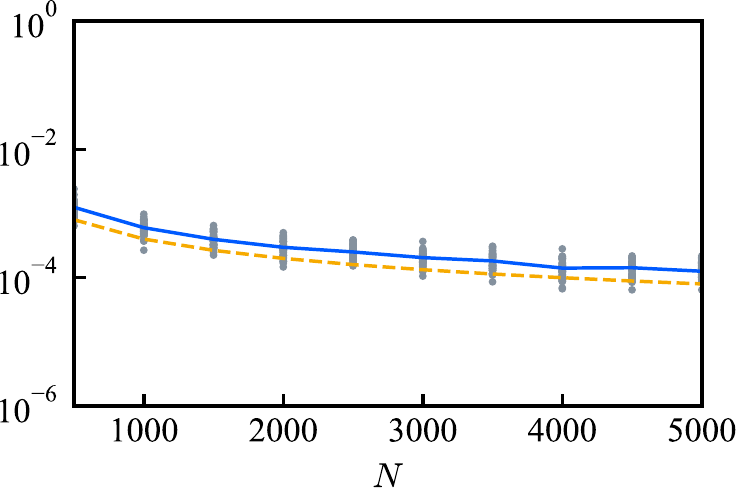}}
    \caption{Comparison between the expectation of estimation error and the sample identifying complexity.}
    \label{fig:variance_ratio}
\end{figure}

Next, we confirm changes in the expectation of estimation error and the sample identifying complexity when the controllability Gramian $\Psi_u$ is varied.
The other Gramian $\Psi_w$ and variance ratio $R_\sigma$ in this simulation are fixed to $2I$ and $100$, respectively.
\figref{fig:gramian} depicts the expectations and sample identifying complexities as with \figref{fig:variance_ratio}.
Additionally, the black dotted lines are the upperbound \eqref{eq:bound} of sample identifying complexities.
The sample identifying complexity in the figure converges to the upperbound as the trace of $\Psi_u$ decreases.
Accordingly, the expectation of estimation error increases, which helps the difficulty of system identification improve.

\begin{figure}[t]
    \centering
    \subfloat[$\Psi_u = 2I$.]{\includegraphics[scale=.6]{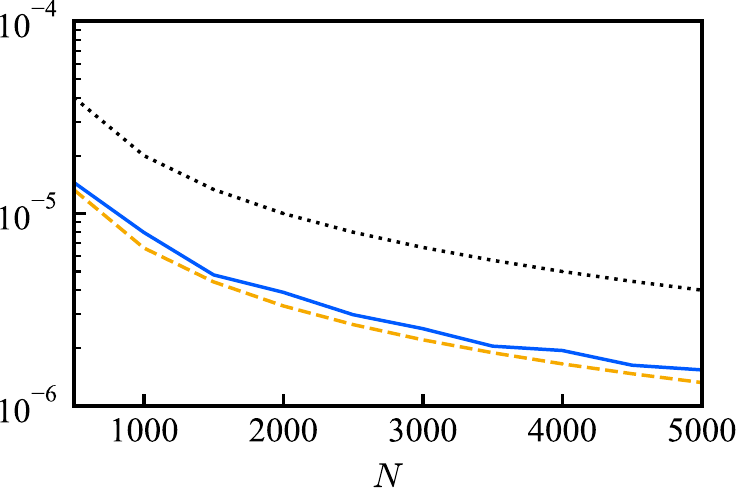}} \quad
    \subfloat[$\Psi_u = I$.]{\includegraphics[scale=.6]{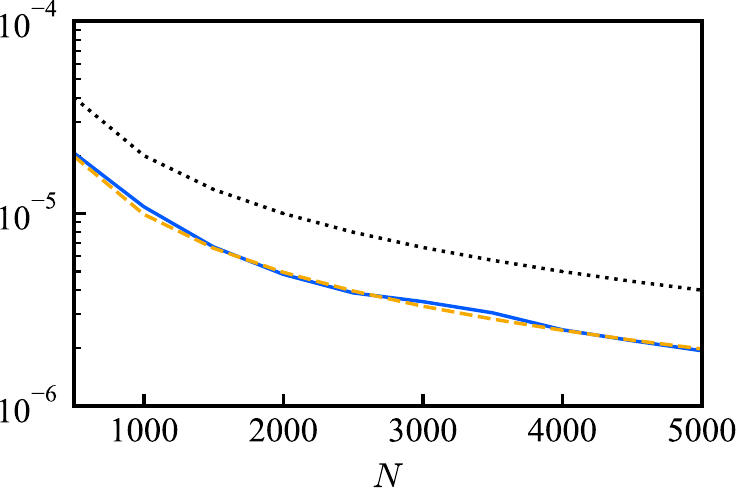}} \quad
    \subfloat[$\Psi_u = 0.5I$.]{\includegraphics[scale=.6]{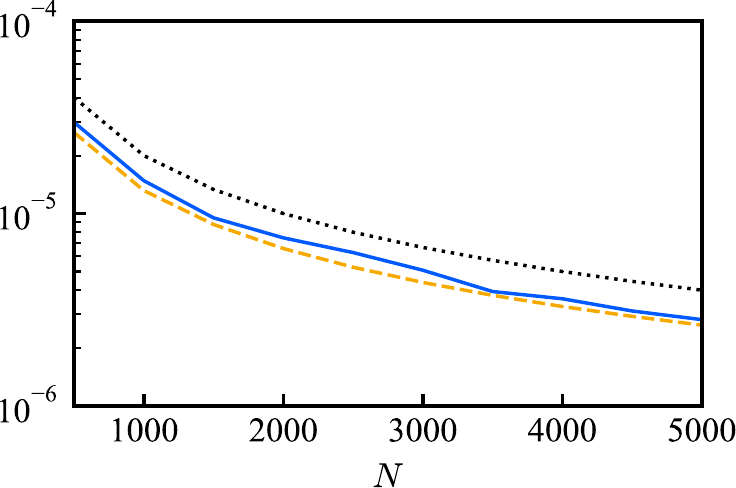}} \\
    \subfloat[$\Psi_u = 0.3I$.]{\includegraphics[scale=.6]{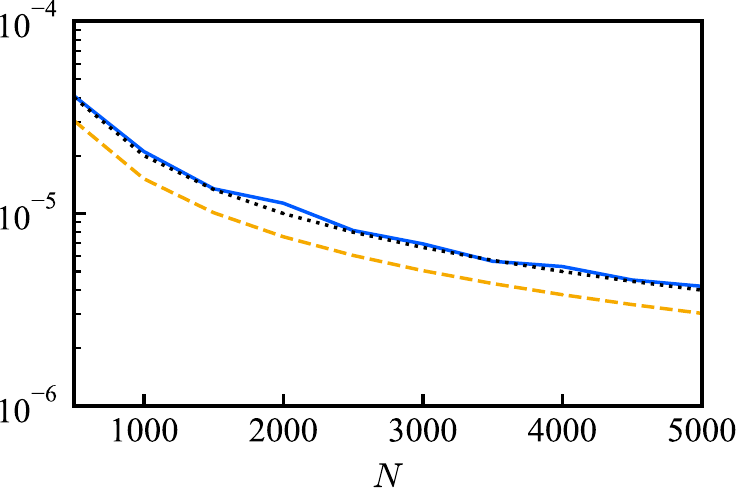}} \quad
    \subfloat[$\Psi_u = 0.2I$.]{\includegraphics[scale=.6]{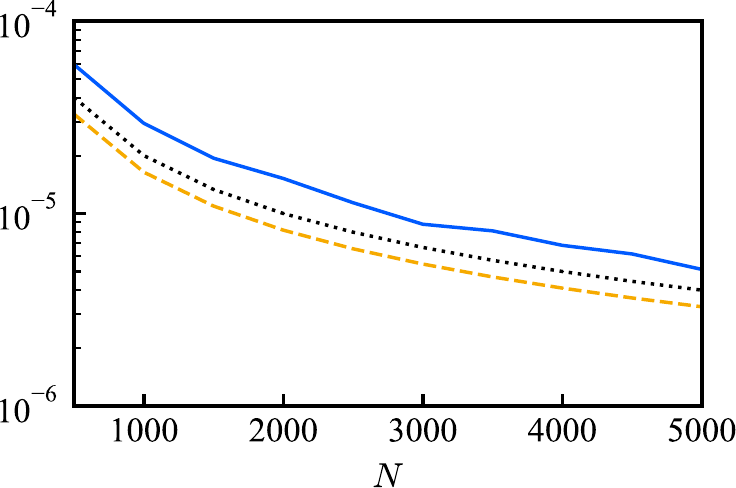}} \quad
    \subfloat[$\Psi_u = 0.1I$.]{\includegraphics[scale=.6]{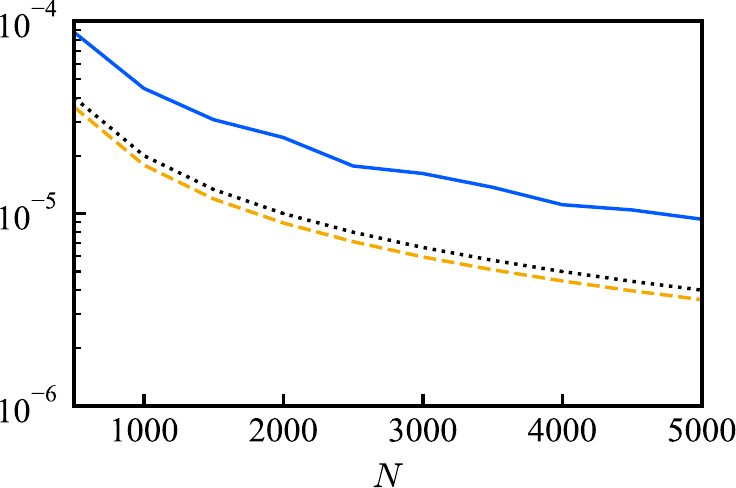}}
    \caption{Changes of the expectation of estimation error and the sample identifying complexity with the various controllability Gramians.}
    \label{fig:gramian}
\end{figure}

Finally, we demonstrate the optimal security parameter design.
Suppose the parameters are $\Psi_w = 2I$, $\Psi_u = 0.5I$, and $R_\sigma = 100$.
Choose the design parameters as $\gamma_c = 10^{-6}$, $\tau_c = 31536 \times 10^4$~s ($10$ years), and $\Upsilon = 442 \times 10^{15}$~FLOPS\footnote{Supercomputer Fugaku. See \url{https://www.top500.org/system/179807/}}.
Then, the minimum sample size $N^\ast$ and optimal sample size $\lambda^\ast$ in \eqref{eq:opt_sec} are given as $13159$ and $74$~bit, respectively.
Moreover, the minimum key length \eqref{eq:opt_key} of modified updatable homomorphic encryption with the algorithms in \exref{ex:dyn_elgamal} and \defref{def:modification}, that guarantees the security of encrypted control system consisting of \eqref{eq:system} and \eqref{eq:ec} in the sense of \defref{def:security}, can be computed as $k^\ast = 712$~bit, where the time complexity of fastest known algorithm for breaking the encryption scheme is $\Omega(k) = \exp\{(64/9)^{1/3} (\ln 2^k)^{1/3} (\ln\ln 2^k)^{2/3}\}$~\cite{Bernstein93}.

\section{Conclusion}
\label{sec:conclusion}

This study presented a modification of a conventional updatable homomorphic encryption scheme for improving the security of encrypted control systems against an eavesdropper and a malicious server.
The novel sample identifying complexity was also proposed under an adversary attempting to identify system parameters in an encrypted control system using a least squares method.
The proposed sample identifying complexity is characterized by controllability Gramians and a variance ratio between an identification input and a system noise.
Furthermore, using the sample identifying complexity, the optimal security parameter for encrypted control systems with the modified updatable homomorphic encryption was designed.
The effectiveness of the proposed method was demonstrated through numerical simulations.

Our future work includes extending the optimal security parameter design under other identification methods, such as subspace identification methods, and considering multi-agent and nonlinear systems.

\section*{Disclosure statement}

The authors report there are no competing interests to declare.

\section*{Funding}

This work was supported by JSPS Grant-in-Aid for JSPS Fellows Grant Number JP21J22442 and JSPS KAKENHI Grant Number JP22H01509.

\bibliographystyle{tfnlm}
\bibliography{encrypted_control_and_optimization,others}

\end{document}